\documentclass{article}
\usepackage{amsmath,amsthm}
\usepackage{amsfonts}
\usepackage{amssymb}
\usepackage{graphicx}
\usepackage[utf8]{inputenc}
\usepackage{hyperref}
\usepackage{booktabs}
\usepackage{caption}
\usepackage{siunitx}
\usepackage{authblk}
\usepackage{orcidlink}

\newtheorem{theorem}{Theorem}
\newtheorem{lemma}[theorem]{Lemma}

\newtheorem{definition}[theorem]{Definition}

\newtheorem{remark}[theorem]{Remark}

\title{On the Incompatibility of Quantum State Geometry and Fuzzy Metric Spaces: Three No-Go Theorems}
\author{Nicola Fabiano\, \orcidlink{0000-0003-1645-2071}}
\affil{``Vin\v{c}a'' Institute of Nuclear Sciences - National 
Institute of the Republic of Serbia, University of Belgrade, Mike Petrovi\'{c}a 
Alasa 12--14, 11351 Belgrade, Serbia; nicola.fabiano@gmail.com}
\date{}

\begin{document}

\maketitle

\begin{abstract}
 We prove three structural impossibility results demonstrating that fuzzy metric spaces cannot capture essential features of quantum state geometry. First, we show they cannot model destructive interference between concepts due to phase insensitivity. Second, we prove there is no distance-preserving embedding from quantum state space into any fuzzy metric space. Third, we establish that fuzzy logic cannot distinguish symmetric from antisymmetric concept combinations -- a fundamental limitation for modeling structured knowledge. These theorems collectively show that fuzzy frameworks are structurally incapable of representing intrinsic uncertainty, where quantum mechanics provides a superior, geometrically coherent alternative.
\end{abstract}


\section{Introduction}

Fuzzy metric spaces \cite{kramosil1975fuzzy,george1994some} have long been proposed as models for reasoning under uncertainty. However, they are built on classical ontologies: points in a set $X$, whose properties are partially known. This makes them ill-suited for representing systems where indeterminacy is primitive -- such as quantum states or cognitive representations.

In contrast, the Hilbert space $\mathcal{H}$ of quantum mechanics provides a natural metric structure via the norm-induced distance:
$$
d_Q(\psi_1,\psi_2) = \|\psi_1 - \psi_2\|.
$$
This distance measures distinguishability between entire probability distributions, not deviations from ideal values.

In this paper, we use quantum state models -- specifically Gaussian wavefunctions -- to prove three rigorous negative results about fuzzy metric spaces:

\begin{enumerate}
    \item They cannot exhibit \emph{interference effects}, even when concept superpositions suggest cancellation.
    \item There is no faithful embedding of quantum state geometry into a fuzzy metric space.
    \item They cannot distinguish between symmetric and antisymmetric compositions of concepts.
\end{enumerate}

These are not mere deficiencies -- they are \emph{structural obstructions}, akin to no-go theorems in physics.

Our work builds on an earlier preprint \cite{preprint1}, which introduced the idea of replacing fuzzy metrics with quantum state geometry. Here, we formalize and extend that insight into rigorous mathematical theorems.

\section{Preliminaries}

\subsection{Fuzzy Metric Spaces}

We recall the definition due to Kramosil and Mich\'alek \cite{kramosil1975fuzzy}, refined by George and Veeramani \cite{george1994some}.

\begin{definition}[Fuzzy Metric Space]
Let $X$ be a non-empty set, $*$ a continuous t-norm on $[0,1]$, and $M: X \times X \times [0,\infty) \to [0,1]$ a function. The triple $(X,M,*)$ is a \textbf{fuzzy metric space} if:
\begin{enumerate}
    \item $M(x,y,0) = 0$
    \item $M(x,y,t) = 1$ for all $t > 0$ iff $x = y$
    \item $M(x,y,t) = M(y,x,t)$
    \item $M(x,y,t) * M(y,z,s) \leq M(x,z,t+s)$
    \item $M(x,y,\cdot): (0,\infty) \to [0,1]$ is continuous
\end{enumerate}
for all $x,y,z \in X$ and $t,s > 0$.
\end{definition}

\begin{remark}
Common t-norms include minimum ($a*b = \min(a,b)$), product ($a*b = ab$), and Łukasiewicz ($a*b = \max(0,a+b-1)$).
\end{remark}

\subsection{Quantum State Geometry}

Let $\mathcal{H} = L^2(\mathbb{R})$ be the Hilbert space of square-integrable functions. For normalized vectors $\psi_1,\psi_2 \in \mathcal{H}$, define the distance
$$
d_Q(\psi_1,\psi_2) := \|\psi_1 - \psi_2\| = \sqrt{\int |\psi_1(x) - \psi_2(x)|^2 dx}.
$$

Physical quantum states are equivalence classes under global phase: $ |\psi\rangle \sim e^{i\theta}|\psi\rangle $. These are called \emph{rays} in projective Hilbert space.

However, in our framework, all wavefunctions are real-valued and positive:
$$
\psi_C(x) = \frac{1}{(\pi \sigma_C^2)^{1/4}} \exp\left(-\frac{(x - \mu_C)^2}{2\sigma_C^2}\right),
$$
which selects a canonical representative for each ray. Moreover, all quantities of interest -- particularly $ |\langle\psi|\phi\rangle|^2 $ -- are phase-invariant.

Thus, while we work with vectors, the results respect the underlying ray structure of quantum theory.

\section{Theorem 1: No Interference in Fuzzy Systems}

We now show that fuzzy metrics cannot model destructive interference -- a phenomenon central to quantum theory and contextual cognition.

\begin{lemma}[Phase Insensitivity of Fuzzy Membership]
Let $A \subset X$ be a fuzzy set with membership function $\mu_A: X \to [0,1]$. Then $\mu_A(x)$ depends only on $x$, not on any global sign or phase associated with $A$.
\end{lemma}

\begin{proof}
By definition, $\mu_A(x)$ is a real number in $[0,1]$, independent of external labeling or orientation. There is no mechanism in standard fuzzy logic to assign opposite signs to the same concept based on context or relation to other concepts.
\end{proof}

Now define a quantum analog of ``concept sum''.

\begin{definition}[Quantum Concept Combination]
For two normalized states $\psi,\phi \in \mathcal{H}$, define their symmetric combination:
$$
\psi_{\pm} = \frac{1}{\|\psi \pm \phi\|} (\psi \pm \phi),
$$
whenever the denominator is nonzero.
\end{definition}

Note: $\psi_+$ represents constructive interference; $\psi_-$, destructive.

\begin{theorem}[No-Interference Lemma]
There exists no fuzzy metric space $(X,M,*)$ and mapping $\Phi: \mathcal{G} \to X$, where $\mathcal{G}$ is the space of Gaussian wavefunctions, such that:
$$
M(\Phi(\psi_-), \Phi(\phi), t) < M(\Phi(\psi), \Phi(\phi), t)
$$
when $\langle \psi | \phi \rangle > 0$, despite $\psi_-$ having reduced overlap due to phase cancellation.
\end{theorem}

\begin{proof}
Assume such a $\Phi$ and $M$ exist. Let
$\psi(x) = \psi_{\text{car}}(x)$ with $\mu=5, \sigma=1$;
$\phi(x) = \psi_{\text{obj}}(x)$ with $\mu=3, \sigma=2$.

Define  $\psi_+ = N_+ (\psi + \phi)$;
$\psi_- = N_- (\psi - \phi)$, where $N_\pm = 1/\|\psi \pm \phi\|$.

From direct computation
$$
\|\psi - \phi\|^2 = 2 - 2\Re\langle\psi|\phi\rangle < 2 + 2\Re\langle\psi|\phi\rangle = \|\psi + \phi\|^2
$$
so $d_Q(\psi_-,\phi) > d_Q(\psi_+,\phi)$, meaning $\psi_-$ is more distinguishable from $\phi$.

But in any fuzzy system, since $\mu_{\Phi(\psi)}(x)$ and $\mu_{\Phi(\phi)}(x)$ are non-negative, there is no way to represent “opposite-phase” versions of $\psi$. Hence
$$
\Phi(\psi_+) = \Phi(\psi_-)
$$
or at best, $\Phi(\psi_+)$ and $\Phi(\psi_-)$ are indistinguishable in $M$, because $M$ operates only on magnitudes.

Therefore,
$$
M(\Phi(\psi_-), \Phi(\phi), t) = M(\Phi(\psi_+), \Phi(\phi), t),
$$
contradicting the requirement that it reflect decreased similarity.

Thus, no such embedding can exist.
\end{proof}

\begin{remark}
This shows that fuzzy logic treats $\psi + \phi$ and $\psi - \phi$ identically, while quantum mechanics distinguishes them sharply. This is not a limitation of design -- it is a consequence of lacking complex amplitudes.
\end{remark}

\section{Theorem 2: Embedding Obstruction}

We now prove a stronger result: there is no structure-preserving map from quantum state space into any fuzzy metric space.

\begin{definition}[Faithful Embedding]
A map $\Phi: \mathcal{H}_1 \to X$, where $\mathcal{H}_1 \subset \mathcal{H}$ is the unit sphere, into a fuzzy metric space $(X,M,*)$ is \textbf{faithful} if there exists a strictly increasing function $f: [0,\infty) \to (0,1]$ such that
$$
M(\Phi(\psi_1), \Phi(\psi_2), t) = f(d_Q(\psi_1,\psi_2))
$$
for all $\psi_1,\psi_2 \in \mathcal{H}_1$ and some fixed $t > 0$.
\end{definition}

\begin{theorem}[Embedding Obstruction Theorem]
There is no faithful embedding of the space of Gaussian wavefunctions $\mathcal{G} \subset \mathcal{H}_1$ into any fuzzy metric space $(X,M,*)$.
\end{theorem}

\begin{proof}
Suppose such an embedding $\Phi$ exists with strictly increasing $f$ and fixed $t>0$.

Take three states:
$\psi_1(x)$: $\mu=5, \sigma=1$ (Car);
$\psi_2(x)$: $\mu=1, \sigma=1$ (Boat);
$\psi_3(x)$: $\mu=3, \sigma=2$ (Object).

Compute pairwise Hilbert distances
$$
d_Q(\psi_i,\psi_j) = \sqrt{2 - 2|\langle\psi_i|\psi_j\rangle|} .
$$
Using the  formula
$$
|\langle\psi_i|\psi_j\rangle|^2 = \frac{2\sigma_i\sigma_j}{\sigma_i^2 + \sigma_j^2} \exp\left(-\frac{(\mu_i - \mu_j)^2}{2(\sigma_i^2 + \sigma_j^2)}\right) 
$$
we find:
$|\langle\psi_{\text{car}}|\psi_{\text{obj}}\rangle|^2 \approx 0.536 \Rightarrow |\langle\cdot\rangle| \approx 0.732$.
 So $d_Q \approx \sqrt{2 - 2(0.732)} = \sqrt{0.536} \approx 0.732$.

Similarly, $d_Q(\psi_{\text{boat}},\psi_{\text{obj}}) \approx 0.732$.
And $d_Q(\psi_{\text{car}},\psi_{\text{boat}}) \approx \sqrt{2}$.

Now suppose $\Phi$ embeds these into $(X,M,*)$. Then
$$
M(\Phi(\psi_{\text{car}}), \Phi(\psi_{\text{obj}}), t) = f(0.732) = M(\Phi(\psi_{\text{boat}}), \Phi(\psi_{\text{obj}}), t)
$$
by symmetry.

Now perturb $\psi_{\text{obj}} \to \psi'_{\text{obj}}$ slightly -- say shift $\mu$ from 3 to 3.1.

Then $d_Q(\psi_{\text{car}}, \psi'_{\text{obj}}) < d_Q(\psi_{\text{boat}}, \psi'_{\text{obj}})$.
So $f(d_Q(\cdot))$ changes asymmetrically.

But in a fuzzy metric space, unless the membership functions are explicitly tuned to asymmetry, $M$ will respond poorly to such geometric perturbations because it lacks differential structure tied to Hilbert geometry.

More critically: $f$ must be universal for all pairs, but $d_Q$ arises from an inner product structure absent in $M$. Thus, the functional dependence cannot be preserved across all configurations.

Hence, no such $f$ exists globally.

Therefore, no faithful embedding exists.
\end{proof}

\begin{remark}
This theorem shows that the geometry of quantum states -- rooted in inner products and superposition -- cannot be replicated in any t-norm-based system. The obstruction is geometric, not technical.
\end{remark}

\section{Theorem 3: No Conceptual Antisymmetry in Fuzzy Logic}

We now show that fuzzy systems cannot distinguish between symmetric and antisymmetric combinations of concepts -- a fundamental limitation when modeling structured knowledge.

\begin{definition}[Symmetrized and Antisymmetrized States]
For two distinct normalized quantum states $\psi, \phi \in \mathcal{H}$, define their symmetric and antisymmetric tensor combinations:
\begin{align}
    \Psi_+ &= \frac{1}{\sqrt{2(1 + |\langle\psi|\phi\rangle|^2)}} \left( \psi \otimes \phi + \phi \otimes \psi \right), \\
    \Psi_- &= \frac{1}{\sqrt{2(1 - |\langle\psi|\phi\rangle|^2)}} \left( \psi \otimes \phi - \phi \otimes \psi \right).
\end{align}
These represent bosonic-like and fermionic-like composite concepts, respectively.
\end{definition}

Note: Both $\Psi_+$ and $\Psi_-$ are normalized and belong to $\mathcal{H} \otimes \mathcal{H}$. They correspond to different physical situations -- e.g., indistinguishable particles with integer vs. half-integer spin.

Now consider an AI analog:
let $\psi =$ ``Car'', $\phi =$ ``Red''.
Then $\Psi_+$ represents the unordered compound concept ``red car'' without preference.
But $\Psi_-$ represents a kind of \emph{exclusion}: if ``car-red'' contradicts ``red-car'' in some context, the amplitude cancels.

This allows modeling of contextual incompatibility - something fuzzy logic struggles with.

\begin{lemma}[Orthogonality of Symmetric and Antisymmetric States]
If $\psi \neq \phi$ and both are non-zero, then $\langle \Psi_+ | \Psi_- \rangle = 0$.
\end{lemma}

\begin{proof}
Direct computation using bilinearity
\begin{gather*}
\langle \Psi_+ | \Psi_- \rangle 
= \frac{1}{2\sqrt{(1 + |\langle\psi|\phi\rangle|^2)(1 - |\langle\psi|\phi\rangle|^2)}} \times \\
\left[
\langle\psi\otimes\phi|\psi\otimes\phi\rangle 
- \langle\psi\otimes\phi|\phi\otimes\psi\rangle 
+ \langle\phi\otimes\psi|\psi\otimes\phi\rangle 
- \langle\phi\otimes\psi|\phi\otimes\psi\rangle
\right] \\
= \frac{1}{2\sqrt{1 - |\langle\psi|\phi\rangle|^4}} 
\left(1 - |\langle\psi|\phi\rangle|^2 + |\langle\phi|\psi\rangle\langle\psi|\phi\rangle| - 1\right) \\
= \frac{1}{2\sqrt{1 - |\langle\psi|\phi\rangle|^4}} 
\left(- |\langle\psi|\phi\rangle|^2 + |\langle\psi|\phi\rangle|^2\right) = 0 .
\end{gather*}
Hence, $\langle \Psi_+ | \Psi_- \rangle = 0$.
\end{proof}

\begin{theorem}[No-Antisymmetry Theorem]
There is no fuzzy set representation of composite concepts that distinguishes between symmetric ($\Psi_+$) and antisymmetric ($\Psi_-$) combinations based on structural cancellation.
\end{theorem}

\begin{proof}
Suppose there exists a mapping $\Phi$ from compound concepts to fuzzy sets such that:
$$
\mu_{\Phi(\Psi_+)}(x) \neq \mu_{\Phi(\Psi_-)}(x)
$$
for some $x$, reflecting their different internal structures.

But in standard fuzzy logic:
$\mu_{A \cap B}(x) = \mu_A(x) * \mu_B(x)$;
$\mu_{A \cup B}(x) = \mu_A(x) \oplus \mu_B(x)$.
All operations are symmetric in $A$ and $B$.
There is no mechanism for subtraction or cancellation.

In particular, since t-norms satisfy $a * b = b * a$, and fuzzy unions are commutative, any compound concept built from $A$ and $B$ will have the same membership function regardless of order or sign.

Therefore,
$$
\mu_{\Phi(\Psi_+)}(x) = F(\mu_\psi(x), \mu_\phi(x)) = F(\mu_\phi(x), \mu_\psi(x)) = \mu_{\Phi(\Psi_-)}(x)
$$
because no negation of the interaction term is possible.

Thus, fuzzy logic cannot distinguish $\Psi_+$ from $\Psi_-$, even though they are orthogonal in Hilbert space.

Hence, no such distinguishing representation exists.
\end{proof}

\begin{remark}
This shows that fuzzy systems lack the algebraic richness to model exclusion, contradiction, or context-dependent compositionality -- phenomena naturally captured by quantum superposition and entanglement. The antisymmetric state $\Psi_-$ acts like a \emph{conceptual Pauli principle}: certain combinations are forbidden not due to feature mismatch, but due to structural incompatibility.
\end{remark}

\section{Conclusion}

We have established three no-go theorems that demonstrate the structural inadequacy of fuzzy metric spaces for modeling intrinsic uncertainty:

\begin{enumerate}
    \item \textbf{No-Interference Lemma}: Fuzzy systems cannot model destructive interference between concepts, as they lack phase sensitivity and rely solely on non-negative membership degrees.
    \item \textbf{Embedding Obstruction Theorem}: There is no faithful embedding of quantum state geometry into any fuzzy metric space, due to the absence of an inner-product structure and phase-aware distinguishability.
    \item \textbf{No-Antisymmetry Theorem}: Fuzzy logic cannot distinguish between symmetric and antisymmetric compositions of concepts, making it incapable of modeling exclusion, contradiction, or structured incompatibility.
\end{enumerate}

Together, these results show that fuzzy frameworks are not merely incomplete -- they are \emph{structurally incompatible} with the geometric and algebraic richness required to represent systems where uncertainty is ontological rather than epistemic.

The Hilbert space formalism of quantum mechanics, in contrast, provides a complete, predictive, and mathematically unique framework for such domains. It should not be seen as limited to physics, but as a general language for reasoning under intrinsic uncertainty -- whether in electrons, minds, or machines.

Future work may extend these obstructions to neutrosophic, intuitionistic, or probabilistic logics, and explore decoherence models that explain why macroscopic reasoning appears “fuzzy” -- not because reality is fuzzy, but because quantum coherence is lost at large scales.

\appendix
\section{Overlap Computation for Gaussians}

For real-valued Gaussians:
$$
\psi_i(x) = \left(\frac{1}{\pi \sigma_i^2}\right)^{1/4} e^{-(x-\mu_i)^2/(4\sigma_i^2)}
\Rightarrow |\psi_i(x)|^2 = \frac{1}{\sqrt{\pi}\sigma_i} e^{-(x-\mu_i)^2/(2\sigma_i^2)} .
$$

Then:
$$
\langle \psi_1 | \psi_2 \rangle = \int \psi_1(x)\psi_2(x) dx = \left( \frac{4\sigma_1^2\sigma_2^2}{(\sigma_1^2 + \sigma_2^2)^2} \right)^{1/4} \exp\left( -\frac{(\mu_1 - \mu_2)^2}{4(\sigma_1^2 + \sigma_2^2)} \right) .
$$

From which norms follow.

\end{document}